\documentclass[11pt]{article}
\usepackage{geometry}
\usepackage{amsmath, amssymb, amsthm}
\usepackage{graphicx}
\usepackage{bm}
\usepackage[linesnumbered,ruled]{algorithm2e}
\usepackage{fullpage}
\newtheorem{theorem}{Theorem}
\newtheorem{lemma}{Lemma}

\usepackage{enumerate}
\graphicspath{{./figure/}}
\newcommand{\II}{\ensuremath{\mathcal{I}}}
\newcommand{\MM}{\ensuremath\mathcal{M}}
\newcommand{\NN}{\ensuremath\mathcal{N}}

\newtheorem{observation}{Observation}

\title{On Reachable Assignments under Dichotomous Preferences\thanks{%
A preliminary version will appear in Proceedings of 
the 24th International Conference on Principles and Practice of Multi-Agent Systems
(PRIMA 2022).
This work was supported by JSPS KAKENHI Grant Numbers 
JP18H04091, JP19K11814, JP20H05793,
JP20H05795, JP21H03397, JP19H05485, JP20K14317,
JP20H05705, JP20K11670, JP18K03391, JP22H05001.}}
\author{Takehiro Ito\thanks{%
Tohoku University, 
{\ttfamily takehiro@tohoku.ac.jp}} \and
Naonori Kakimura\thanks{%
Keio University,  
{\ttfamily kakimura@math.keio.ac.jp}} \and
Naoyuki Kamiyama\thanks{%
Kyushu University, 
{\ttfamily kamiyama@imi.kyushu-u.ac.jp}} \and  
Yusuke Kobayashi\thanks{%
Kyoto University, 
{\ttfamily yusuke@kurims.kyoto-u.ac.jp}%
} \and
Yuta Nozaki\thanks{%
Hiroshima University, 
{\ttfamily nozakiy@hiroshima-u.ac.jp}%
} \and
Yoshio Okamoto\thanks{%
The University of Electro-Communications, 
{\ttfamily okamotoy@uec.ac.jp}
} \and
Kenta Ozeki\thanks{%
Yokohama National University, 
{\ttfamily ozeki-kenta-xr@ynu.ac.jp}}}
\date{}

\begin{document}

\maketitle

\begin{abstract}
We consider the problem of determining 
whether a target item assignment can be reached from 
an initial item assignment by a sequence of pairwise exchanges of 
items between agents. 
In particular, we consider the situation where 
each agent has a dichotomous preference over the items, 
that is, each agent evaluates each item as acceptable 
or unacceptable. 
Furthermore, we assume that 
communication between agents is limited, and 
the relationship is represented by an undirected graph. 
Then, a pair of agents can exchange their items only if 
they are connected by an edge and the involved items are acceptable. 
We prove that 
this problem is $\mathsf{PSPACE}$-complete 
even when the communication graph is complete  (that is, 
every pair of agents can exchange their items), and 
this problem can be solved in polynomial time if an input graph 
is a tree. 
\end{abstract} 

\section{Introduction}

\subsection{Our Contributions}
We consider the following problem.
We are given a set of agents and a set of items.
There are as many items as agents.
Each agent has a \emph{dichotomous preference} over the items,
that is, each agent evaluates each item as acceptable or unacceptable. 
(See, e.g., \cite{BM04} 
for situations where dichotomous preferences naturally 
arise.)
Over the set of agents, we are given a communication graph.
We are also given two assignments of items to agents, where
each agent receives an acceptable item.
Now, we want to determine whether one assignment can be reached from the other assignment by rational exchanges.
Here, a rational exchange means that each of the two agents
accepts the item assigned to the other, and
they are joined by an edge in the communication graph.

We investigate algorithmic aspects of this problem.
Our results are two-fold.
We first prove that 
our problem can be 
solved in polynomial time if the communication graph is a tree.
Second, we prove that 
our problem is $\mathsf{PSPACE}$-complete 
even when the communication graph is complete (that is, every 
pair of agents can exchange their items).
This $\mathsf{PSPACE}$-completeness result shows an 
interesting contrast to 
the $\mathsf{NP}$-completeness in the strict preference case~\cite{MB20}. 

The question studied in this paper is related to the generation of a random assignment.
Bogomolnaia and Moulin~\cite{BM04} stated several good properties of random assignments in situations with dichotomous preferences.
One of the typical methods for generating a random assignment is based on the Markov chain Monte Carlo method \cite{LP17}.
In this method, we consider a sequence of small changes for assignments and hope that the resulting assignment is sufficiently random.
For this method to work, we require all possible assignments can be reached from an arbitrary initial assignment, i.e., the irreducibility of the Markov chain.
This paper studies such an aspect of random assignments under dichotomous preferences from the perspective of combinatorial reconfiguration \cite{N18}.

\subsection{Backgrounds}

The problem of assigning indivisible items to agents has been extensively studied 
in algorithmic game theory and computational 
social choice (see, e.g., \cite{M13,KlausMR16}). 
Applications of this kind of problem 
include 
job allocation, college admission, school choice, 
kidney exchange, and junior doctor allocation to 
hospital posts. 
When we consider this kind of problem, we implicitly assume that 
agents can observe the situations of all the agents 
and freely communicate with others. 
Recently, assignment problems without these assumptions 
have been studied. 
For example, fairness concepts based on limited 
observations on others
have been considered in \cite{AKP17,ABCGL18,BCGLMW19,BKN18,FMT19}. 
In a typical setting in this direction, we are given 
a graph defined on the agents and fairness properties are defined on 
a pair of agents joined by an edge of 
the graph or the neighborhoods of vertices. 
This paper is concerned with the latter assumption, that is, 
we consider assignment problems in the situation
where the communication between agents is limited.  

Our problem is concerned with situations where  
each agent is initially endowed with a single item:
Those situations commonly arise in the housing market problem~\cite{SS74}.
In the housing market problem, 
the goal is to reach one of the desired item assignments by exchanging 
items among agents from the initial assignment. 
For example, the top-trading cycle algorithm proposed by 
Shapley and Scarf~\cite{SS74} is one of the most fundamental 
algorithms for this problem, and 
variants of the top-trading cycle algorithm have been 
proposed (see, e.g., \cite{AS99,AD12}). 
As said above, in the standard housing market problem, we assume that any pair of agents can 
exchange their items. 
However, in some situations, this assumption does not seem to be 
realistic. For example, when we consider trading among a large
number of agents, it is natural to consider that agents can 
exchange their items only if they can communicate with 
each other. 
Recently, the setting with restricted exchanges has been considered
\cite{GLW17,HX20,LPS21,MB20}. 
More precisely, 
we are given an undirected graph defined on the agents 
representing possible exchanges, and 
a pair of agents can exchange their items only if 
they are joined by an edge.

Gourv\`{e}s, Lesca, and Wilczynski~\cite{GLW17} initiated 
the algorithmic research of exchanges over social networks in the housing market problem.
They assumed that each agent has a strict preference 
over the items, and considered the question that asks which allocation of the items can emerge
by rational exchanges between two agents. 
More concretely, they considered the problem of determining 
whether a target assignment can be reached from 
an initial assignment by rational exchanges between two agents. 
Here a rational exchange means that
both agents prefer the item assigned to the other to 
her/his currently assigned item and they are joined by an edge. 
We can see that 
if the target assignment is reachable from the initial 
assignment, then the target assignment can emerge by 
decentralized rational trades between agents. 
Gourv\`{e}s, Lesca, and Wilczynski~\cite{GLW17}
proved that 
this problem is $\mathsf{NP}$-complete in general, and 
can be solved in polynomial time
when the communication graph is a tree. 
Later,  M\"{u}ller and Benter~\cite{MB20} 
proved that this problem is $\mathsf{NP}$-complete even 
when the communication graph is complete, and 
can be solved in polynomial time 
when the communication graph is a cycle. 

In addition to reachability between assignments 
by rational exchanges, the problem of determining 
whether an assignment where a specified agent 
receives a target item can be reached from an initial 
assignment by rational exchanges has been studied.
Gourv\`{e}s, Lesca, and Wilczynski~\cite{GLW17} proved that 
this problem is $\mathsf{NP}$-complete even when the communication graph
is a tree.  
Huang and Xiao~\cite{HX20} proved that 
this problem can be solved in polynomial time 
when the communication graph is a path. 
In addition, they proved the $\mathsf{NP}$-completeness 
and the polynomial-time solvability in stars 
for preferences that may contain ties. 

Li, Plaxton, and Sinha~\cite{LPS21}
considered the following variant of the model mentioned above~\cite{GLW17,HX20,MB20}.
In their model, we are given a graph defined on 
the items and an exchange between some agents is allowed 
if their current items are joined by an edge. 
For this model, Li, Plaxton, and Sinha~\cite{LPS21}
proved similar results to the results for the former model \cite{GLW17,HX20,MB20}.

Our problem can be regarded as one kind of 
problems where we are given an initial configuration and 
a target configuration of some combinatorial objects, and 
the goal is to check the reachability between these two 
configurations via some specified operations.
In theoretical computer science,
this kind of problem has been studied 
under the name of \emph{combinatorial reconfiguration}. 
The algorithmic studies of combinatorial reconfiguration 
were initiated by Ito et al.~\cite{IDHPSUU11}.
See, e.g., \cite{N18} for a survey of combinatorial 
reconfiguration. 
In Section~\ref{setion:pspace_complete}, we use a known 
result in combinatorial reconfiguration. 

\section{Preliminaries}

Assume that we are given a finite set $N$ of agents and 
a finite set $M$ of items such that $|N| = |M|$.
For each item $j \in M$, we are given a subset $N_j$ of agents
who can accept $j$.
For each agent $i \in N$, define a subset $M_i \subseteq M$ as the set of acceptable items in $M$, i.e.,
$j \in M_i$ if and only if $i \in N_j$.
For a subset $X \subseteq M$, we define
$N_X = \bigcup_{j \in X}N_j$.
We define the ordered families $\MM$ and $\NN$ as
$\MM = (M_i \mid i \in N)$ and $\NN = (N_j \mid j \in M)$.
Furthermore, we are given an undirected graph $G = (N, E)$.

The setup can be rephrased in terms of graphs.
From the family $\NN=(N_j \mid j \in M)$, we may define the
following bipartite graph $H$.
The vertex set of $H$ is $N\cup M$, and two vertices $i\in N$ and $j \in M$ are joined by an edge if and only if $i \in N_j$ (or equivalently, $j \in M_i$).
The graph $G$ is defined over the set $N$.
See \figurename~\ref{fig:itemalloc_graph1}.

\begin{figure}[t]
\centering
\includegraphics[scale=1]{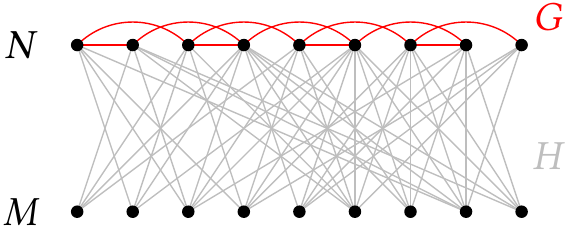}
\caption{The graph representation. Graph $G$ is shown in red, and graph $H$ is shown in gray.}
\label{fig:itemalloc_graph1}
\end{figure}

A bijection $a \colon N \to M$ is called an \emph{assignment}
if $a(i) \in M_i$ for every agent $i \in N$, i.e., 
$a(i)$ is an item that is acceptable for $i$.
By the assignment $a$, we say an item $j$ is \emph{assigned} to
an agent $i$ if $a(i)=j$.
In terms of the graph $H$, an assignment corresponds to a \emph{perfect matching} of $H$.
Hall's marriage theorem states that a perfect matching of $H$ exists if and only if $|S| \leq |N_S|$ for all $S\subseteq M$.
Hall's marriage theorem will be used in the next section to prove our theorems.

For a pair of assignments $a, b \colon N \to M$, 
we write $a \to b$ if 
there exist distinct agents $i, i^{\prime} \in N$ 
satisfying the following two conditions.
\begin{itemize}
\item 
For every agent $k \in N \setminus \{i,i^{\prime}\}$, 
$a(k) = b(k)$. 
\item 
$a(i) = b(i^{\prime})$, 
$a(i^{\prime}) = b(i)$, and 
$\{i,i^{\prime}\} \in E$.
\end{itemize}
See \figurename~\ref{fig:itemalloc_swap1}.
As a handy notation, we use $a(Y) = \{a(i) \mid i \in Y\}$ for every $Y \subseteq N$ and $a^{-1}(X) = \{a^{-1}(j) \mid j \in X\}$ for every $X \subseteq M$.

\begin{figure}[t]
\centering
\includegraphics[scale=1]{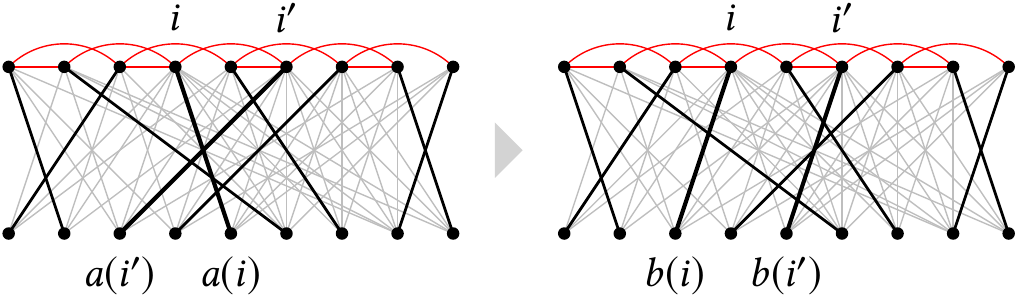}
\caption{An exchange operation. Assignments are drawn with thick black segments as perfect matchings.}
\label{fig:itemalloc_swap1}
\end{figure}

Our problem is defined as follows. 
An instance is specified by a $6$-tuple 
$\II = (N, M, \NN, G, a, b)$,
where $a$ and $b$ are assignments.
The goal is to determine whether 
there exists a sequence $a_0, a_1, \dots, a_{\ell}$ of assignments
such that 
$a_{t -1} \to a_t$ 
for every integer $t \in \{1,2,\ldots,\ell\}$,
$a_0 = a$, and $a_{\ell} = b$. 
In this case, we say that $a$ can be \emph{reconfigured} to $b$, or $b$ is \emph{reachable} from $a$.
Observe that $a_0^{-1}(j), a_1^{-1}(j), \dots , a_\ell^{-1}(j)$ are in the same connected component of $G[N_j]$, where $G[N_j]$ is the subgraph of $G$ induced by $N_j$. 
Thus, when we consider the reachability of the assignments, we may assume that $G[N_j]$ is connected for every $j \in M$ without loss of generality. 

For the family $\NN$, 
a non-empty subset $X \subseteq M$ of items is \emph{stable} if
$|X| = \left|N_X\right|$. 
We remind that $N_X = \bigcup_{j \in X}N_j$.
A stable subset $X \subseteq M$ is \emph{proper} 
if $\emptyset \not= X \subsetneq M$.

\section{Trees: A Characterization}

In this section, we consider the case when $G$ is a tree.  
We give a sufficient condition for the reachability of the assignments, which is essential to design a polynomial-time algorithm in Section~\ref{sec:tree:alg}.
As described in the previous section, it suffices to deal with the case when $G[N_j]$ is connected for every $j \in M$.

\begin{theorem}\label{thm:chartree}
Suppose that $G$ is a tree and $G[N_j]$ is connected for every $j \in M$. 
If there exists no proper stable subset of items in $M$, then 
every assignment can be reconfigured to any other assignment. 
\end{theorem}

We prove the theorem by induction on $|N|$. 
When $|N|=1$, the claim is obvious. 

Consider an instance $(N, M, \NN, G, a, b)$ with $|N| \ge 2$. 
Assume that there exists no proper stable subset of items in $M$, i.e., 
$|N_X| \ge |X|+1$ for any nonempty subset $X \subsetneq M$. 
We consider the following two cases separately: 
\begin{enumerate}
\item
There exists a subset $X \subseteq M$ such that $N_X \not= N$ and $|N_X| = |X|+1$. 
\item
For any nonempty subset $X \subseteq M$, we have that $N_X=N$ or $|N_X| \ge |X|+2$. 
\end{enumerate}

\subsection{Case 1}

Suppose that there exists a subset $X \subseteq M$ such that $N_X \not= N$ and $|N_X| = |X|+1$. 
Among such sets, let $X$ be an inclusionwise minimal one. 
Note that $X \not= M$. 

\begin{lemma} 
\label{clm:01}
$G[N_X]$ is connected. 
\end{lemma} 
\begin{proof}
Assume to the contrary that $G[N_X]$ is not connected. 
Then, there exists a partition $X_1, \cdots , X_t$ of $X$ with $t \ge 2$ such that 
$G[N_{X_1}], \dots , G[N_{X_t}]$ are distinct connected components of $G[N_X]$. 
Since there exists no proper stable subset, we obtain
$|N_{X_i}| > |X_i|$ for $i=1, \dots , t$. 
Hence,  
$|N_X| = \sum  |N_{X_i}| \ge \sum (|X_i|+1) \ge |X| + t >  |X|+1$, which is a contradiction. 
\end{proof}

We denote $R := N_X$ to simplify the notation. 
The idea is to consider the inside of $G[R]$ and the graph obtained from $G$ by shrinking $R$, separately. 

Since $|R| = |X|+1$, we observe the following. 
\begin{observation}
\label{obs:01}
For any assignment $c \colon N \to M$, 
there exists an item $j \in M \setminus X$ such that $c(R) = X \cup \{j\}$. 
\end{observation}

For an item $j \in M \setminus X$, 
a bijection $c' \colon R \to X \cup \{j\}$ is called an {\em assignment in $R$ using $j$}
if $c'(i) \in M_i$ for any $i \in R$. 
If $j$ is clear from the context, it is simply called an {\em assignment in $R$}. 

\begin{lemma} 
\label{clm:02}
Let $j$ be an item in $M \setminus X$ and let $i$ be an agent in $N_j \cap R$. 
Then, there exists an assignment $c'$ in $R$ such that $c'(i) = j$. 
\end{lemma} 

\begin{proof}
It suffices to show the existence of an appropriate bijection from $R \setminus \{i\}$ to $X$. 
For any nonempty subset $S \subseteq X$, we obtain 
$|N_S| \ge |S|+1$ as there exists no proper stable set.
This shows that $|S| \le |N_S \setminus \{ i \}|$ holds for all $S \subseteq X$. 
Therefore, a desired assignment $c'$ exists by Hall's marriage theorem. 
\end{proof}

\begin{lemma} 
\label{clm:03}
Let $j$ be an item in $M \setminus X$. 
Define $N' := R$, $M' := X \cup \{j\}$, and 
$\NN' := (N_{j'} \cap R \mid j' \in X \cup \{j\})$. 
If $|N_j \cap R| \ge 2$, then   
 $(N', M', \NN', G[R], a', b')$ is a yes-instance (i.e., $a'$ can be reconfigured to $b'$) for any assignments $a'$ and $b'$ in $R$.
\end{lemma} 

\begin{proof}
We first show that $|N'_Y| \ge |Y|+1$ for any nonempty subset $Y \subsetneq M'$, 
where  $N'_Y := N_Y \cap R$, by the following case analysis. 
\begin{itemize}
\item
Suppose that $j \not\in Y$. In this case, $|N'_Y| = |N_Y| \ge |Y|+1$ holds as $M$ has no proper stable subset. 
\item
Suppose that $Y = X' \cup \{j\}$ holds for some nonempty subset $X' \subsetneq X$. 
Since $|N_{X'}| \ge |X'| + 2$ by the minimality of  $X$, we obtain 
$|N'_Y| = |N_Y \cap R| \ge |N_{X'} \cap R| = |N_{X'}| \ge |X'| + 2 \ge |Y|+1$. 
\item
Suppose that $Y = \{j\}$. 
In this case, $|N'_Y| = |N_j \cap R| \ge 2 = |Y|+1$ by the assumption. 
\end{itemize}
Therefore, we obtain $|N'_Y| \ge |Y|+1$ for each case. 
We also see that $G[N_{j'} \cap R]$ is connected for each $j' \in X \cup \{j\}$, because $G[N_{j'}]$ and $G[R]$ are connected (see 
Lemma~\ref{clm:01}) 
and $G$ is a tree. 
Since $|N'| < |N|$, by applying the induction hypothesis, we see that $(N', M', \NN', G[R], a', b')$ is a yes-instance. 
\end{proof}

By using these lemmas, 
we have the following. 

\begin{lemma} 
\label{clm:04}
Let $j$ be an item in $M \setminus X$. 
Define $N' := R$, $M' := X \cup \{j\}$, and 
$\NN' := (N_{j'} \cap R \mid j' \in X \cup \{j\})$. 
Let $i_1, i_2 \in N_j \cap R$ be agents and let $a'$ be an assignment in $R$
such that $a'(i_1) = j$. 
Then, there exists an assignment $b'$ in $R$ such that $b'(i_2) = j$ 
and $(N', M', \NN', G[R], a', b')$ is a yes-instance (i.e., $a'$ can be reconfigured to $b'$). 
\end{lemma} 

\begin{proof}
If $i_1 = i_2$, then $b'=a'$ satisfies the condition. 
Otherwise, since $|N_j \cap R| \ge |\{i_1, i_2\}| = 2$, the lemma  
holds by Lemmas~\ref{clm:02} and~\ref{clm:03}. 
\end{proof}

The following lemma shows that any assignment can be reconfigured to an assignment for which we can apply Lemma~\ref{clm:03} 
in $G[R]$. 

\begin{lemma} 
\label{clm:05}
Let $c \colon N \to M$ be an assignment. 
Then, there exist an assignment $c^*\colon N \to M$ and 
an item $j^* \in M \setminus X$ such that $c^*(R) = X \cup \{j^*\}$, $|N_{j^*} \cap R| \ge 2$, and 
$c$ can be reconfigured to $c^*$. 
\end{lemma} 

\begin{proof}
By Observation~\ref{obs:01}, there exists a unique vertex $q$ in $R$ such that $c(q) \in M \setminus X$. 
Let $Q \subseteq N$ be the vertex set of the connected component of $G - E(R)$ containing $q$, where $E(R)$ is the set of edges with both endpoints in $R$. 
Since $G$ is a tree, we obtain the following: 
\begin{enumerate}[(C1)]
    \item $q$ is a cut vertex of $G$ separating $Q\setminus \{q\}$ and $N \setminus Q$;
    \item Any vertex in $N \setminus Q$ that is adjacent to $q$ is contained in $R$. 
\end{enumerate}
Define $Y \subseteq c(Q)$ as an inclusionwise minimal nonempty set of items such that  $|N_Y \cap Q| = |Y|$. 
Note that such $Y$ exists, because $Y = c(Q)$ satisfies that $|N_Y \cap Q| = |Y|$. 
We observe a few properties of $Y$. 
First, $G[N_Y \cap Q]$ is connected by the minimality of $Y$.  
Second, by $Y \subseteq c(Q)$ and $|N_Y \cap Q| = |Y|$, it holds that $c^{-1}(Y) = N_Y \cap Q$. 
Third, since $Y$ is not a proper stable subset, we obtain $|N_Y| > |Y| = |N_Y \cap Q|$, and hence 
there exists a vertex in $N_Y \setminus Q$. 
Then, there exists an item $j^* \in Y$ with $N_{j^*} \setminus Q \not= \emptyset$. We also see that $N_{j^*} \cap Q \not= \emptyset$ as $c^{-1}(j^*) \in Q$. Since $G[N_{j^*}]$ is connected and $N_{j^*}$ intersects both $Q$ and $N \setminus Q$, (C1) shows that $N_{j^*}$ contains $q$. Furthermore, $N_{j^*}$ contains a vertex $q' \in N \setminus Q$ that is adjacent to $q$. Since $q' \in R$ by (C2), we obtain $|N_{j^*} \cap R| \ge |\{q, q'\}| = 2$.

We next claim 
that there exists a bijection $c^* \colon c^{-1}(Y) \to Y$ such that 
$c^*(i) \in M_i$ for $i \in c^{-1}(Y)$ and $c^*(q) = j^*$. 
For any nonempty subset $S \subseteq Y \setminus \{j^*\}$, we obtain 
$|N_S \cap c^{-1}(Y)| = |N_S \cap Q| \ge |S|+1$ by the minimality of $Y$.
This shows that $|S| \le |(N_S \cap c^{-1}(Y)) \setminus \{ q \}|$ holds for all $S \subseteq Y \setminus \{j^*\}$. 
Therefore, a desired bijection $c^*$ exists by Hall's marriage theorem. 
Note that $c^*$ can be naturally extended to a bijection from $N$ to $M$
by defining $c^*(i) = c(i)$ for $i \in N \setminus c^{-1}(Y)$. 
Then, it holds that $c^*(R) = X \cup \{j^*\}$. 

Finally, we show that $c$ can be reconfigured to $c^*$. 
To see this, it suffices to consider $G[c^{-1}(Y)]$. 
For any nonempty subset $S \subsetneq Y$, we obtain 
$|N_S \cap c^{-1}(Y)|  = |N_S \cap Q| \ge |S|+1$ by the minimality of $Y$.
This means that there is no proper stable subset if we restrict the instance to $G[c^{-1}(Y)]$. 
We also see that $G[N_{j'} \cap c^{-1}(Y)]$ is connected for each $j' \in Y$, because $G[N_{j'}]$ and $G[c^{-1}(Y)]=G[N_Y \cap Q]$ are connected and $G$ is a tree. 
Therefore, by the induction hypothesis, any pair of assignments in $G[c^{-1}(Y)]$ can be reconfigured to each other. 
This shows that $c$ can be reconfigured to $c^*$. 
\end{proof}

By applying Lemma~\ref{clm:05} 
in which $c=b$, 
there exist an assignment $b^*\colon N \to M$ and 
an item $j^* \in M \setminus X$ such that $b^*(R) = X \cup \{j^*\}$, $|N_{j^*} \cap R| \ge 2$, and 
$b$ can be reconfigured to $b^*$. 
Conversely, it is obvious that $b^*$ can be reconfigured to $b$. 

Let $G^\circ$ be the graph obtained from $G$ by shrinking $R$ to a single vertex $r$, and let $N^\circ$ be its vertex set, i.e., $N^\circ = (N \setminus R) \cup \{r\}$. 
Let $M^\circ = M \setminus X$. 
For $j \in M^\circ$, define $N^\circ_j$ as follows: 
$$
N^\circ_j=
\begin{cases}
N_j \cup \{r\}& \mbox{if $N_j \cap R \not= \emptyset$,} \\
N_j & \mbox{otherwise.} 
\end{cases}
$$
We can easily see that $G^\circ[N^\circ_j]$ is connected for each $j \in M^\circ$ as $G[N_j]$ is connected. 
For assignments $a$ and $b^*$ in $G$, let $a^\circ$ and $b^\circ$ be the corresponding assignments in $G^\circ$, 
which are naturally defined by Observation~\ref{obs:01}.

\begin{lemma} 
\label{clm:06}
$(N^\circ, M^\circ, \NN^\circ, G^\circ, a^\circ, b^\circ)$ is a yes-instance. 
\end{lemma} 

\begin{proof}
We show that this instance has no proper stable subset of items.
Assume to the contrary that $Y \subsetneq M^\circ$ is a proper stable subset, that is, $|N^\circ_Y| = |Y|$. 
If $r \not\in N^\circ_Y$, then $|N_Y| = |N^\circ_Y| = |Y|$, and hence $Y$ is a proper stable subset in the original instance, which is a contradiction. 
Otherwise, $|N_{Y\cup X}| = |(N^\circ_Y \setminus \{r\}) \cup R| = |N^\circ_Y| - 1 + |R| = |Y| + |X|$, 
and hence $Y \cup X$ is a proper stable subset in the original instance, which is a contradiction. 
Therefore, $(N^\circ, M^\circ, \NN^\circ, G^\circ, a^\circ, b^\circ)$ has no proper stable subset of items, which shows that 
it is a yes-instance by the induction hypothesis. 
\end{proof}

We next show that a reconfiguration in $G^\circ$ can be converted to one in $G$ in the following sense. 

\begin{lemma} 
\label{clm:10}
Let $c^\circ_1, c^\circ_2 \colon N^\circ \to M^\circ$ be assignments in $G^\circ$ such that $c^\circ_1 \rightarrow c^\circ_2$, and 
let $c_1 \colon N\to M$ be an assignment in $G$ that corresponds to $c^\circ_1$. 
Then, there exists an assignment $c_2 \colon N\to M$ in $G$ 
such that $c_2$ corresponds to $c^\circ_2$ and $c_1$ can be reconfigured to $c_2$ in $G$.
\end{lemma} 
\begin{proof}
Suppose that $c^\circ_1(i) = c^\circ_2(i^{\prime})$, 
$c^\circ_1(i^{\prime}) = c^\circ_2(i)$, and 
$\{i,i^{\prime}\} \in E(G^\circ)$.

We first consider the case when $r \not\in \{i,i^{\prime}\}$. 
Define $c_2 \colon N\to M$ as 
$c_2(i) = c_1(i^{\prime})$, 
$c_2(i^{\prime}) = c_1(i)$, and 
$c_2(k) = c_1(k)$ for $k \in N \setminus \{i,i^{\prime}\}$.  
Then, $c_2$ corresponds to $c^\circ_2$ and $c_1 \rightarrow c_2$.

We next consider the case when $r \in \{i,i^{\prime}\}$.
By symmetry, we may assume that  $r=i'$. 
Let $j = c^\circ_1(r)$
and let $q \in R$ be the unique vertex that is adjacent to $i$ in $G$. 
Since $c^\circ_1(r) = c^\circ_2(i) = j$ implies that $N_j \cap R \not= \emptyset$ and $i \in N_j$, it holds that $q \in N_j$. 
By using Lemma~\ref{clm:04} 
in which $a'$ is the restriction of $c_1$ to $R$ and $i_2=q$, 
we see that there exists an assignment $c_3\colon N\to M$ in $G$ 
such that 
$c_3(q)=j$, 
$c_3(k) = c_1(k)$ for $k \in N \setminus R$, and 
$c_1$ can be reconfigured to $c_3$. 
Define $c_2 \colon N\to M$ as 
$c_2(i) = c_3(q)$, 
$c_2(q) = c_3(i)$, and 
$c_2(k) = c_3(k)$ for $k \in N \setminus \{i,q\}$.  
Then, $c_2$ corresponds to $c^\circ_2$ and $c_3 \rightarrow c_2$, 
which shows that $c_2$ satisfies the conditions in the lemma.
\end{proof}

We are now ready to show that $(N, M, \NN, G, a, b)$ is a yes-instance. 
Since Lemma~\ref{clm:06} 
shows that $(N^\circ, M^\circ, \NN^\circ, G^\circ, a^\circ, b^\circ)$ is a yes-instance, there exists a reconfiguration sequence from $a^\circ$ to $b^\circ$. 
By using Lemma~\ref{clm:10}, 
this sequence can be converted to a reconfiguration sequence
from $a$ to some assignment $b'$ in $G$ such that 
$b'(i) = b^\circ (i) = b^*(i)$ for $i \in N \setminus R$ and $b'(R) = X \cup \{b^\circ (r)\} = X \cup \{j^*\}$. 
Furthermore, since $|N_{j^*} \cap R| \ge 2$, Lemma~\ref{clm:03} 
shows that $b'$ can be reconfigured to $b^*$. 
Therefore, there exists a reconfiguration sequence
$a \to \dots \to b' \to \dots \to  b^* \to \dots \to b$, and hence 
$(N, M, \NN, G, a, b)$ is a yes-instance. 

\subsection{Case 2}

In this subsection, we consider the case when $N_X=N$ or $|N_X| \ge |X|+2$ holds 
for any nonempty subset $X \subseteq M$. 
We begin with the following lemmas. 

\begin{lemma} 
\label{clm:08}
If $a(\ell) = b(\ell)$ for some leaf $\ell$, then $a$ can be reconfigured to $b$. 
\end{lemma} 
\begin{proof}
Consider the instance $(N', M', \NN', G', a, b)$ obtained from $(N, M, \NN, G, a, b)$ by removing $\ell$ and $a(\ell)$. 
That is, $G' = G - \ell$, $N' = N \setminus \{\ell \}$, $M' = M \setminus \{a(\ell)\}$, $N'_j = N_j \setminus \{\ell\}$ for $j \in M'$, and 
the domains of $a$ and $b$ are restricted to $N'$. 
Then, for any nonempty subset $Y \subsetneq M'$, we obtain $|N'_Y| = |N_Y \setminus \{\ell\}| \ge |N_Y| - 1 \ge (|Y| + 2) - 1 \ge |Y| + 1$, 
where we note that  $ |N_Y| \ge \min (|N|, |Y| + 2) = |Y| + 2$ by the assumption in this subsection. 
Therefore, the obtained instance has no proper stable subset, and hence the restriction of $a$ can be reconfigured to that of $b$ in $G'$ by the induction hypothesis. 
Since $a(\ell) = b(\ell)$, this shows that $a$ can be reconfigured to $b$ in $G$. 
\end{proof}

\begin{lemma} 
\label{clm:11}
If there exist distinct leaves $\ell$ and $\ell'$ such that $a(\ell') \not= b(\ell)$, then
$a$ can be reconfigured to $b$. 
\end{lemma} 

\begin{proof}
We first show that there exists an assignment $c\colon N \to M$ such that $c(\ell') = a(\ell')$ and $c(\ell) = b(\ell)$. 
For any nonempty subset $S \subseteq M \setminus \{a(\ell'), b(\ell)\}$, we obtain 
$|N_S \setminus \{ \ell, \ell'\}| \ge |N_S| - 2 \ge (|S|+2) - 2 = |S|$, 
where we note that  $|N_S| \ge \min (|N|, |S| + 2) = |S| + 2$ by the assumption in this subsection. 
Therefore, a desired assignment $c$ exists by Hall's marriage theorem.  

Since $a(\ell')=c(\ell')$, Lemma~\ref{clm:08} 
shows that $a$ can be reconfigured to $c$. 
Similarly, since $c(\ell) = b(\ell)$, $c$ can be reconfigured to $b$ by 
Lemma~\ref{clm:08} 
again. 
Therefore, $a$ can be reconfigured to $b$, which completes the proof. 
\end{proof}

We are now ready to show that $a$ can be reconfigured to $b$. 
If $G$ has at least three leaves, then 
there exist distinct leaves $\ell$ and $\ell'$ such that $a(\ell') \not= b(\ell)$, 
and hence $a$ can be reconfigured to $b$ by 
Lemma~\ref{clm:11}. 

Thus, the remaining case is when $G$ is a path with exactly two leaves $\ell$ and $\ell'$. 
We may assume that $a(\ell) = b(\ell')$ and $a(\ell')=b(\ell)$, 
since otherwise $a$ can be reconfigured to $b$ by 
Lemma~\ref{clm:11}. 
We may also assume that $G$ has at least three vertices, since otherwise 
the 
lemma 
is obvious. 
Let $q$ be the unique vertex adjacent to $\ell$. 

We now show that there exists an assignment $c\colon N \to M$ such that $c(\ell) = a(\ell)$ and $c(q) = a(\ell')$. 
Note that $q \in N_{a(\ell')}$, because $a(\ell')=b(\ell)$ and $G$ is a path. 
For any nonempty subset $S \subseteq M \setminus \{a(\ell), a(\ell')\}$, we obtain 
$|N_S \setminus \{ q, \ell\}| \ge |N_S| - 2 \ge (|S|+2) - 2 = |S|$ by the assumption in this subsection. 
Therefore, a desired assignment $c$ exists by Hall's marriage theorem.  

Since $a(\ell) = c(\ell)$, $a$ can be reconfigured to $c$ by 
Lemma~\ref{clm:08}. 
Furthermore, since $c(\ell') \not= c(q) = a(\ell') = b(\ell)$, $c$ can be reconfigured to $b$ by 
Lemma~\ref{clm:11}. 
By combining them, we have that $a$ can be reconfigured to $b$, which completes the proof.

\section{Trees: Algorithm}
\label{sec:tree:alg}

Theorem \ref{thm:chartree} leads to the following polynomial-time algorithm to determine whether two given assignments can be reconfigured to each other.

\begin{theorem}
\label{thm:tree-algo}
We can determine in polynomial time whether
for a given instance $(N, M, \NN, G, a, b)$, $a$ can be reconfigured to $b$, when $G$ is a tree.
\end{theorem}

Recall that we may assume that $G[N_j]$ is connected for every $j \in M$. 
To prove Theorem \ref{thm:tree-algo}, we first give a polynomial-time algorithm to find a proper stable subset of items, if it exists.

\begin{lemma}
\label{lem:propstab-algo}
We can determine in polynomial time whether
for a given instance $(N, M, \NN, G, a, b)$, there exists a proper stable subset of items and find one with minimum size 
if it exists, 
when $G$ is a tree.
\end{lemma}

Below we present a proof for Lemma~\ref{lem:propstab-algo} using submodular functions.
Before the proof, we summarize definitions and properties of submodular functions that we use in the proof.

For a finite set $\Xi$, the \emph{power set} of $\Xi$ is the family of all subsets of $\Xi$ and denoted by $2^\Xi$.
A function $f\colon 2^\Xi \to \mathbb{R}$ is \emph{submodular} if $f(X)+f(Y) \geq f(X\cup Y)+f(X\cap Y)$ for all $X, Y \subseteq \Xi$. 
The submodular function minimization is a problem to find a set $X^* \subseteq \Xi$ such that $f(X^*) \leq f(X)$ for all $X \subseteq \Xi$; such a set $X^*$ is a \emph{minimizer} of $f$.
Here, the submodular function $f$ is not given explicitly, but it is given as oracle access.
Namely, we assume that we may retrieve the value $f(X)$ for each set $X \subseteq \Xi$ in polynomial time.

A minimizer of a submodular function $f$ does not have to be unique.
If $X^*$ and $Y^*$ are minimizers of $f$, then $X^* \cup Y^*$ and $X^* \cap Y^*$ are also minimizers of $f$, which can easily be seen from the submodularity of $f$.
This implies that there exists a unique minimum-size minimizer of any submodular function.
A minimum-size minimizer of a submodular function (given as oracle access) can be obtained in polynomial time~\cite{murota-dcabook}.

\begin{proof}[Proof of Lemma~\ref{lem:propstab-algo}]
For each item $j \in M$, we define the function $f_j\colon 2^{M\setminus \{j\}} \to \mathbb{R}$ as
\[
f_j(X) = |N_{X\cup \{j\}}| - |X\cup \{j\}|
\]
for all $X\subseteq M\setminus \{j\}$.
Since $H$ has the assignment $a$, $f_j(X) \geq 0$ for all $X \subseteq M \setminus \{j\}$ by Hall's marriage theorem.
Thus, since $f_j(M\setminus \{j\}) = 0$, the minimum value of $f_j$ is zero. 
Notice that 
$f_j(X) =0$ if and only if 
$X\cup \{j\}$ is stable.

It is easy to see that the function $f_j$ is submodular, and for any submodular function, a unique minimum-size minimizer can be found in polynomial time as noted above. 
Let $X_j$ be the unique minimum-size minimizer of $f_j$ and
let $X^*_j = X_j \cup \{j\}$.
Then, $X^*_j$ is the unique minimum-size stable 
subset containing $j$. 

Let $j^* \in M$ be an item that minimizes $|X^*_{j^*}|$. 
Since $X^*_j$ is the unique minimum-size stable 
subset containing $j$ for each $j \in M$, 
$X^*_{j^*}$ is the minimum-size nonempty stable subset of items. 
Therefore, a proper stable subset exists if and only if $X^*_{j^*} \not= M$, which can be determined in polynomial time by computing $X^*_{j^*}$. 
Furthermore, if $X^*_{j^*} \not= M$, then $X^*_{j^*}$ is 
a proper stable subset with minimum size.
\end{proof}

For our algorithm, we first decide whether, for a given instance $(N, M, \NN, G, a, b)$, there exists a proper stable subset.
If none exists, then Theorem \ref{thm:chartree} implies that $a$ can be reconfigured to $b$, and we are done.
Assume that there exists a proper stable subset of items for the instance.
Let $X$ be one with minimum size.

We first observe that, by the minimality, $G[N_X]$ is connected.
To see this, assume to the contrary that $G[N_X]$ is not connected. 
Let $(Y_1, \dots , Y_p)$ be the partition of $X$ such that $G[N_{Y_t}]$ forms a connected component of $G[N_X]$ for each $t \in \{1, \dots , p\}$, where $p \ge 2$. 
Note that such a partition exists, because $G[N_{j'}]$ is connected for all $j' \in X$. 
Since $X$ is a minimum-size proper stable set, it holds that $|N_{Y_t}|>|Y_t|$ for all $t \in \{1, \dots , p\}$.
This implies that $|N_X| - |X| = \sum_{t=1}^p (|N_{Y_t}| - |Y_t|) > 0$, which is a contradiction.

We then apply our algorithm recursively to the instances obtained by $G[N_X]$ and $G[N\setminus N_X]$, respectively.
Here, $G[N\setminus N_X]$ consists of several connected components, whose vertex sets are denoted by $N^1,\dots, N^\ell$, for some $\ell \geq 1$, and $G[N\setminus N_X]$ yields $\ell$ instances.

The following lemma is crucial.
For $i=1,\dots, \ell$, define $M^i = a(N^i)$.
\begin{lemma}
\label{lem:propstab-impossible}
Let $(N, M, \NN, G, a, b)$ be an instance such that 
$G$ is a tree and
let $X$ be a proper stable subset of items.
If there exists an item $j \in M^i$ such that
$b^{-1}(j) \not\in N^i$,
then $a$ cannot be reconfigured to $b$.
\end{lemma}
\begin{proof}
For simplicity, we may assume that $j$ is in $M^1$ and $b^{-1}(j) \not\in N^1$.
Since $G$ is a tree, there exists a unique edge $(i_1, i'_1)$ between $N^1$ and $N_X$, where $i_1\in N^1$ and $i'_1\in N_X$.

Suppose that $a$ can be reconfigured to $b$
by a reconfiguration sequence
$a=a_0 \to a_1 \to \dots \to a_\ell=b$.
Then, in the reconfiguration sequence, there exists an index $t$ such that $a^{-1}_{t-1}(j) = i_1$ and $a^{-1}_{t}(j) = i'_1$.
That is, $j=a_{t-1}(i_1)=a_{t}(i'_1)$. 
This means that there exists an item $j^\prime \in M$ such that
$j^\prime = a_{t-1}(i'_1) = a_{t}(i_1)$,
i.e., from the assignment $a_{t-1}$ to $a_{t}$, the agents
$i_1$ and $i'_1$ exchange the items $j$ and $j^\prime$.
Since $X$ is stable, we see that $a^{-1}_{t-1}(X)=N_X$, and hence $j^\prime \in X$ holds.
Since $j^\prime = a_{t}(i_1)$, 
$i_1\in N_{j^\prime} \subseteq N_X$.
This contradicts that $i_1$ is in $N^1$.
\end{proof}

Armed with Lemmas \ref{lem:propstab-algo} and \ref{lem:propstab-impossible}, we are ready for describing our algorithm.

\begin{enumerate}[Step 1.]
\item Decide whether a proper stable subset exists. If there is none, then we answer Yes.
Otherwise, let $X$ be a proper stable subset with minimum size, 
and proceed to Step~2.
\item The subgraph $G[N\setminus N_X]$ consists of several connected components, whose vertex sets are denoted by $N^1,\dots, N^\ell$, for some $\ell \geq 1$.
For $i=1,\dots, \ell$, define $M^i = a(N^i)$.
Check whether there exists an item $j \in M^i$ such that
$b^{-1}(j) \not\in N^i$.
If there exists such an item, then we answer No.
Otherwise, proceed to Step 3.
\item We construct $\ell+1$ smaller instances as follows.
The first instance is $(N_X, X, \NN_X, G[N_X], a_X, b_X)$, where
$\NN_X = (N_j \mid j \in X)$ and $a_X, b_X\colon N_X \to X$ are the restrictions of $a, b$ to $N_X$, respectively. 
The other instances are $(N^i, M^i, \NN^i, G[N^i], a_i, b_i)$ for $i=1,\dots, \ell$, where
$\NN^i = (N_j \cap N^i \mid j \in M^i)$ and
$a_i, b_i\colon N^i \to M^i$ are the restrictions of $a, b$ to $N^i$, respectively.
By the assumption of Step 3, those instances are well-defined.
Those $\ell + 1$ instances are solved recursively.
If the answers to the smaller instances are all Yes, then the answer to the whole instance is also Yes.
Otherwise, the answer to the whole instance is No.
\end{enumerate}

The correctness is immediate from Theorem \ref{thm:chartree} and Lemma \ref{lem:propstab-impossible}, and the running time is polynomial by Lemma \ref{lem:propstab-algo}.
Thus, the proof of Theorem \ref{thm:tree-algo} is completed.

\section[Complete Graphs: PSPACE-Completeness]{Complete Graphs: $\mathsf{PSPACE}$-Completeness}
\label{setion:pspace_complete}

In this section, we prove that 
our problem is $\mathsf{PSPACE}$-complete 
even when $G$ is a complete graph.

\begin{theorem}
The problem is $\mathsf{PSPACE}$-complete even if $G$ is a complete graph.
\end{theorem}
\begin{proof}
The membership in $\mathsf{PSPACE}$ is immediate since each assignment can be encoded in polynomial space, and each swap can be performed in polynomial space (even in polynomial time). Thus, we concentrate on $\mathsf{PSPACE}$-hardness.

The following ``bipartite perfect matching reconfiguration problem'' is known to be  $\mathsf{PSPACE}$-complete~\cite{BonamyBHIKMMW19} .
We are given a bipartite graph $H^{\prime}$ and two perfect matchings $M_1, M_2$ of $H^{\prime}$, and we are asked to decide whether $M_1$ can be transformed to $M_2$ by a sequence of exchanges of
two matching edges with two non-matching edges
such that those four edges form a cycle of $H^{\prime}$.

From an instance $(H^{\prime}, M_1, M_2)$ of the bipartite perfect matching reconfiguration problem, we construct an instance $(N, M, \NN, G, a, b)$ of our problem where $G$ is a complete graph.

Denote two color classes (partite sets) of $H^{\prime}$ by $A$ and $B$.
Then, let $N=A$ and $M=B$.
Since $M_1, M_2$ are perfect matchings of $H^{\prime}$, it holds that
$|N|=|A|=|B|=|M|$.
For each $j \in B=M$, we define $N_j$ as the set of vertices in $A=N$ that are adjacent to $j$ in $H^{\prime}$.
Then, $\NN$ is the family $(N_j \mid j \in M)$.
The assignments $a,b$ are defined by $M_1, M_2$ as
$a(i)=j$ if and only if $\{i,j\} \in M_1$ and
$b(i)=j$ if and only if $\{i,j\} \in M_2$.
The graph $G$ is a complete graph on $N$.
This finishes the construction of the instance.
We emphasize that $G[N_j]$ is indeed connected for every $j \in M$ since $G$ is a complete graph.

Observe that an exchange operation in the bipartite perfect matching reconfiguration problem precisely corresponds to an exchange operation in our problem.
Thus, the reduction is sound and complete, and the proof is finished.
\end{proof}

Note that the bipartite perfect matching reconfiguration problem is $\mathsf{PSPACE}$-complete even when the input graph has a bounded bandwidth and maximum degree five~\cite{BonamyBHIKMMW19}.

For strict preferences, the problem in complete graphs is $\mathsf{NP}$-complete~\cite{MB20}.
Thus, we encounter a huge difference between the complexity status for dichotomous preferences ($\mathsf{PSPACE}$-complete) and strict preferences ($\mathsf{NP}$-complete).
This is because with strict preferences each exchange strictly improves the utility of the two agents involved in the exchange, and thus the length of a reconfiguration sequence is always bounded by a polynomial of the number of agents.
On the other hand, with dichotomous preferences, a reconfiguration sequence can be exponentially long.

\section{Concluding Remarks}

Further studies are required for the following research directions.
The complexity status for other types of graphs $G$ is not known.
The shortest length of a reconfiguration sequence is not known even for trees.
In particular, when there is a reconfiguration sequence, we do not know whether the shortest length is bounded by a polynomial in $|N|$. 
We may also study other types of preferences.

\bibliographystyle{plain}
\bibliography{arxiv_prima_v2}

\end{document}